\documentclass[twoside]{article}%
\usepackage{amssymb}
\usepackage{amsfonts}
\usepackage{amsmath}
\usepackage{graphicx}%
\providecommand{\U}[1]{\protect\rule{.1in}{.1in}}
\newtheorem{theorem}{Theorem}

\newtheorem{corollary}[theorem]{Corollary}

\newtheorem{lemma}[theorem]{Lemma}

\newtheorem{remark}[theorem]{Remark}

\newenvironment{proof}[1][Proof]{\noindent\textbf{#1.} }{\ \rule{0.5em}{0.5em}}
\begin{document}

\title{\textbf{Self-Similar Blowup Solutions to the 2-Component Degasperis-Procesi
Shallow Water System}}
\author{M\textsc{anwai Yuen\thanks{E-mail address: nevetsyuen@hotmail.com }}\\\textit{Department of Applied Mathematics, The Hong Kong Polytechnic
University,}\\\textit{Hung Hom, Kowloon, Hong Kong}}
\date{Revised 13-Aug-2010}
\maketitle

\begin{abstract}
In this article, we study the self-similar solutions of the 2-component
Degasperis-Procesi water system:%
\begin{equation}
\left\{
\begin{array}
[c]{c}%
\rho_{t}+k_{2}u\rho_{x}+(k_{1}+k_{2})\rho u_{x}=0\\
u_{t}-u_{xxt}+4uu_{x}-3u_{x}u_{xx}-uu_{xxx}+k_{3}\rho\rho_{x}=0.
\end{array}
\right.
\end{equation}
By the separation method, we can obtain a class of self-similar solutions,%
\begin{equation}
\left\{
\begin{array}
[c]{c}%
\rho(t,x)=\max(\frac{f\left(  \eta\right)  }{a(4t)^{(k_{1}+k_{2})/4}},\text{
}0),\text{ }u(t,x)=\frac{\overset{\cdot}{a}(4t)}{a(4t)}x\\
\overset{\cdot\cdot}{a}(s)-\frac{\xi}{4a(s)^{\kappa}}=0,\text{ }a(0)=a_{0}%
\neq0,\text{ }\overset{\cdot}{a}(0)=a_{1}\\
f(\eta)=\frac{k_{3}}{\xi}\sqrt{-\frac{\xi}{k_{3}}\eta^{2}+\left(  \frac{\xi
}{k_{3}}\alpha\right)  ^{2}}%
\end{array}
\right.
\end{equation}
where $\eta=\frac{x}{a(s)^{1/4}}$ with $s=4t;$ $\kappa=\frac{k_{1}}{2}%
+k_{2}-1,$ $\alpha\geq0,$ $\xi<0$, $a_{0}$ and $a_{1}$ are constants.\newline
which the local or global behavior can be determined by the corresponding
Emden equation. The results are very similar to the one obtained for the
2-component Camassa-Holm equations. Our analytical solutions could provide
concrete examples for testing the validation and stabilities of numerical
methods for the systems. With the characteristic line method, blowup
phenomenon for $k_{3}\geq0$ is also studied.

Mathematics Subject Classification (2010): 35B40, 35B44, 35C06, 35Q53

Key Words: 2-Component Degasperis-Procesi\textbf{,} Shallow Water System,
Analytical Solutions, Blowup, Global, Self-Similar, Separation Method,
Construction of Solutions, Moving Boundary, 2-Component Camassa-Holm Equations

\end{abstract}

\section{Introduction}

The 2-component Degasperis-Procesi shallow water system \cite{Po} and
\cite{JG}, can be expressed in the following form%
\begin{equation}
\left\{
\begin{array}
[c]{c}%
\rho_{t}+k_{2}u\rho_{x}+(k_{1}+k_{2})\rho u_{x}=0,\text{ }x\in R\\
u_{t}-u_{xxt}+4uu_{x}-3u_{x}u_{xx}-uu_{xxx}+k_{3}\rho\rho_{x}=0.
\end{array}
\right.  \label{2com}%
\end{equation}
Here $k_{1}$, $k_{2}$, $k_{3}$ are constants. $u=u(x,t)\in R$ is the velocity
of fluid. And $\rho=\rho(t,x)\geq0$ is the density of fluid. For $\rho=0,$ the
system returns to the Degasperis-Procesi equation \cite{ELY}, \cite{LY} and
\cite{Zhou}.

In this article, we adopt an alternative approach (method of separation) to
study some self-similar solutions of 2-component Camassa-Holm equations
(\ref{2com}). Indeed, we observe that the isentropic Euler, Euler-Poisson,
Navier-Stokes and Navier-Stokes-Poisson systems are written by:
\begin{equation}
\left\{
\begin{array}
[c]{rl}%
{\normalsize \rho}_{t}{\normalsize +\nabla\cdot(\rho\vec{u})} &
{\normalsize =}{\normalsize 0}\\
{\normalsize \rho\lbrack\vec{u}}_{t}+\left(  {\normalsize \vec{u}\cdot\nabla
}\right)  {\normalsize \vec{u})]+\nabla P} & {\normalsize =-}{\normalsize \rho
\nabla\Phi+vis(\rho,\vec{u})}\\
{\normalsize \Delta\Phi(t,x)} & {\normalsize =\alpha(N)}{\normalsize \rho}%
\end{array}
\right.  \label{Euler-Poisson}%
\end{equation}
where $\alpha(N)$ is a constant related to the unit ball in $R^{N}$:
$\alpha(1)=2$; $\alpha(2)=2\pi$ and for $N\geq3,$%
\begin{equation}
\alpha(N)=N(N-2)Vol(N)=N(N-2)\frac{\pi^{N/2}}{\Gamma(N/2+1)}%
\end{equation}
where $Vol(N)$ is the volume of the unit ball in $R^{N}$ and $\Gamma$ is a
Gamma function. And as usual, $\rho=\rho(t,\vec{x})$ and $\vec{u}=\vec
{u}(t,\vec{x})\in\mathbf{R}^{N}$ are the density and the velocity
respectively. $P=P(\rho)=K\rho^{r}$\ is the pressure, the constant $K\geq0$
and $\gamma\geq1$. And ${\normalsize vis(\rho,\vec{u})}$ is the viscosity
function.\newline We may seek the radial solutions
\begin{equation}
\rho(t,\vec{x})=\rho(t,r)\text{ and }\vec{u}=\frac{\vec{x}}{r}V(t,r)=:\frac
{\vec{x}}{r}V
\end{equation}
with $r=\left(  \sum_{i=1}^{N}x_{i}^{2}\right)  ^{1/2}$.\newline By the
standard computation, the Euler equations in radial symmetry can be written in
the following form:%
\begin{equation}
\left\{
\begin{array}
[c]{c}%
\rho_{t}+V\rho_{r}+\rho V_{r}+\dfrac{N-1}{r}\rho V=0\\
\rho\left(  V_{t}+VV_{r}\right)  +K\frac{\partial}{\partial r}\rho^{\gamma}=0.
\end{array}
\right.  \label{eqeq1}%
\end{equation}
For the mass equation in radial symmetry, (\ref{eqeq1})$_{1}$, we well know
the solutions' structure (Lemma 3, \cite{Yuen 2}):%
\begin{equation}
\rho(t,r)=\frac{f(\frac{r}{a(t)})}{a(t)^{N}},\text{ }{\normalsize u(t,r)=}%
\frac{\overset{\cdot}{a}(r)}{a(r)}r{\normalsize .}%
\end{equation}

As the 2-component Degasperis-Procesi equations (\ref{2com}), are very similar
to the Euler system (\ref{eqeq1}), we can apply the separation method
(\cite{GW}, \cite{M1}, \cite{Li}, \cite{Yuen1}, \cite{Yuen 2}) to the systems
(\ref{2com}). We note that for the 2-component Camassa-Holm equations
\cite{I}, \cite{CH}, \cite{GY} and \cite{Guo1},%
\begin{equation}
\left\{
\begin{array}
[c]{c}%
\rho_{t}+u\rho_{x}+\rho u_{x}=0\\
m_{t}+2um_{x}+um_{x}+\sigma\rho\rho_{x}=0
\end{array}
\right.
\end{equation}
with
\begin{equation}
m=u-\alpha^{2}u_{xx}.
\end{equation}
By the separation method, we can obtain a class of blowup or global solutions
for $\sigma=1$ or $-1$ \cite{Yuen 3}. In particular, for the integrable system
with $\sigma=1$, we have the global solutions:%
\begin{equation}
\left\{
\begin{array}
[c]{c}%
\rho(t,x)=\left\{
\begin{array}
[c]{c}%
\frac{f\left(  \eta\right)  }{a(3t)^{1/3}},\text{ for }\eta^{2}<\xi\alpha
^{2}\\
0,\text{ for }\eta^{2}\geq\xi\alpha^{2}%
\end{array}
\right.  ,u(t,x)=\frac{\overset{\cdot}{a}(3t)}{a(3t)}x\\
\overset{\cdot\cdot}{a}(s)-\frac{\xi}{3a(s)^{1/3}}=0,\text{ }a(0)=a_{0}%
>0,\text{ }\overset{\cdot}{a}(0)=a_{1}\\
f(\eta)=\frac{1}{\xi}\sqrt{-\xi\eta^{2}+\left(  \xi\alpha\right)  ^{2}}%
\end{array}
\right.  \label{caca}%
\end{equation}
where $\eta=\frac{x}{a(s)^{1/3}}$ with $s=3t;$ $\xi>0$ and $\alpha\geq0$ are
arbitrary constants.\newline As the two component of Camassa-Holm equations
are very similar to the two component of Degasperis-Procesi equations, In this
article, we continue to apply the separation method, to deduce the nonlinear
partial differential equations into much simpler ordinary differential
equations. Therefore, we can contribute a new class of self-similar solutions
in the following corresponding result:

\begin{theorem}
\label{thm:1}We define the function $a(s)$ is the solution of the Emden
equation:
\begin{equation}
\left\{
\begin{array}
[c]{c}%
\overset{\cdot\cdot}{a}(s)-\frac{\xi}{4a(s)^{\kappa}}=0\\
a(0)=a_{0}>0,\text{ }\overset{\cdot}{a}(0)=a_{1}%
\end{array}
\right.  \text{ } \label{Emden}%
\end{equation}
and
\begin{equation}
f(\eta)=\frac{k_{3}}{\xi}\sqrt{-\frac{\xi}{k_{3}}\eta^{2}+\left(  \frac{\xi
}{k_{3}}\alpha\right)  ^{2}}%
\end{equation}
where $\eta=\frac{x}{a(s)^{k_{2}/4}}$ with $s=4t;$ $\kappa=\frac{k_{1}}%
{2}+k_{2}-1,$ $\alpha\geq0,$ $\xi\neq0$, $a_{0}$ and $a_{1}$ are
constants.\newline For the 2-component Degasperis-Procesi equations
(\ref{2com}), there exists a family of solutions, those are:\newline(1)for
$k_{3}=0,$ and $\xi=0,$ or%
\begin{equation}
\rho(t,x)=\frac{\rho_{0}(\eta)}{a(4t)^{(k_{1}+k_{2})/4}},\text{ }%
u(t,x)=\frac{\overset{\cdot}{a}(4t)}{a(4t)}x
\end{equation}
where $\rho_{0}\geq0,$ is an arbitrary $C^{1}$ function.\newline(2) for
$k_{3}>0$ and $\xi>0,$ or \newline(3) for $k_{3}<0$ and $\xi<0,$%
\begin{equation}
\rho(t,x)=\max(\frac{f\left(  \eta\right)  }{a(4t)^{(k_{1}+k_{2})/4}},\text{
}0)\text{ },\text{ }u(t,x)=\frac{\overset{\cdot}{a}(4t)}{a(4t)}x.
\label{ChCh2}%
\end{equation}

\end{theorem}

\begin{remark}
The structure of solutions (\ref{ChCh2}) of the 2-component Degasperis-Procesi
equations are very similar the one (\ref{caca}), \cite{Yuen 3} of the
Camassa-Holm equations.
\end{remark}

\section{Separation Method}

The mass of the solution is not conserved except on $k_{1}=0$ and $k_{2}=1$ in
equation (\ref{2com})$_{1}$. However, we can also design a nice functional
structure for the mass equation:

\begin{lemma}
\label{lem:generalsolution}For the $1$-dimensional equation of mass
(\ref{2com})$_{1}$:
\begin{equation}
\rho_{t}+k_{2}u\rho_{x}+(k_{1}+k_{2})\rho u_{x}=0 \label{generalmass}%
\end{equation}
there exist solutions,%
\begin{equation}
\rho(t,x)=\frac{f(\frac{x}{a(4t)^{k_{2}/4}})}{a(4t)^{(k_{1}+k_{2})/4}},\text{
}{\normalsize u(t,x)=}\frac{\overset{\cdot}{a}(4t)}{a(4t)}x \label{functional}%
\end{equation}
with the form $f(\eta)\geq0\in C^{1}$ with $\eta=\frac{x}{a(4t)^{k_{2}/4}}$,
and $a(4t)>0\in C^{1}.$
\end{lemma}

\begin{proof}
We just plug (\ref{functional}) into (\ref{generalmass}) to check:
\begin{align}
&  \rho_{t}+k_{2}u\rho_{x}+(k_{1}+k_{2})\rho u_{x}\\[0.1in]
&  =\frac{\partial}{\partial t}\left(  \frac{f(\frac{x}{a(4t)^{k_{2}/4}}%
)}{a(4t)^{(k_{1}+k_{2})/4}}\right)  +k_{2}\frac{\overset{\cdot}{a}(4t)}%
{a(4t)}x\frac{\partial}{\partial x}\left(  \frac{f(\frac{x}{a(4t)^{k_{2}/4}}%
)}{a(4t)^{(k_{1}+k_{2})/4}}\right)  +(k_{1}+k_{2})\frac{f(\frac{x}%
{a(4t)^{k_{2}/4}})}{a(4t)^{(k_{1}+k_{2})/4}}\frac{\partial}{\partial x}\left(
\frac{\overset{\cdot}{a}(4t)}{a(4t)}x\right) \\[0.1in]
&  =\frac{1}{a(4t)^{(k_{1}+k_{2})/4+1}}(-\frac{(k_{1}+k_{2})}{4})\cdot
\overset{\cdot}{a}(4t)\cdot4\cdot f(\frac{x}{a(4t)^{k_{2}/4}})+\frac
{1}{a(4t)^{(k_{1}+k_{2})/4}}\overset{\cdot}{f}(\frac{x}{a(4t)^{k_{2}/4}}%
)\frac{\partial}{\partial t}(\frac{x}{a(4t)^{k_{2}/4}})\\[0.1in]
&  +\frac{k_{2}\overset{\cdot}{a}(4t)x}{a(4t)}\frac{\overset{\cdot}{f}%
(\frac{x}{a(4t)^{k_{2}/4}})}{a(4t)^{(k_{1}+k_{2})/4}}\frac{\partial}{\partial
x}\left(  \frac{x}{a(4t)^{k_{2}/4}}\right)  +\frac{(k_{1}+k_{2})f(\frac
{x}{a(4t)^{k_{2}/4}})}{a(4t)^{(k_{1}+k_{2})/4}}\frac{\overset{\cdot}{a}%
(4t)}{a(4t)}\\[0.1in]
&  =-\frac{(k_{1}+k_{2})\overset{\cdot}{a}(4t)f(\frac{x}{a(4t)^{k_{2}/4}}%
)}{a(4t)^{(k_{1}+k_{2})/4+1}}-\frac{1}{a(4t)^{(k_{1}+k_{2})/4}}\overset{\cdot
}{f}(\frac{x}{a(4t)^{k_{2}/4}})\frac{x}{a(4t)^{k_{2}/4+1}}\frac{k_{2}}{4}%
\dot{a}(4t)\cdot4\\[0.1in]
&  +\frac{k_{2}\overset{\cdot}{a}(4t)x}{a(4t)}\frac{\overset{\cdot}{f}%
(\frac{x}{a(4t)^{k_{2}/4}})}{a(4t)^{(k_{1}+k_{2})/4}}\frac{1}{a(4t)^{k_{2}/4}%
}+(k_{1}+k_{2})\frac{f(\frac{x}{a(4t)^{k_{2}/4}})}{a(4t)^{(k_{1}+k_{2})/4}%
}\frac{\overset{\cdot}{a}(4t)}{a(4t)}\\[0.1in]
&  =0.
\end{align}
The proof is completed.
\end{proof}

On the other hand, in \cite{DXY} and \cite{Y1}, the qualitative properties of
the Emden equation%
\begin{equation}
\left\{
\begin{array}
[c]{c}%
\ddot{a}(s)-\frac{\xi}{a(s)^{\kappa}}=0\\
\text{ }a(0)=a_{0}\neq0,\text{ }\dot{a}(0)=a_{1}%
\end{array}
\right.
\end{equation}
where for $\kappa\geq1$ in \cite{DXY} and \cite{Y1} and $\kappa=\frac{1}{3}$
in \cite{Yuen 3}; \newline were studied. Therefore, the similar local
existence of the Emden equations (\ref{Emden}),
\begin{equation}
\left\{
\begin{array}
[c]{c}%
\ddot{a}(s)-\frac{\xi}{a(s)^{\kappa}}=0\\
a(0)=a_{0}\neq0,\text{ }\dot{a}(0)=a_{1}%
\end{array}
\right.
\end{equation}
can be proved by the standard fixed point theorem \cite{DXY} and \cite{Y1}. To
additionally show the blowup property of the time function $a(s)$, the
following lemmas are needed.

\begin{lemma}
\label{lemma22}For the Emden equation (\ref{Emden}),%
\begin{equation}
\left\{
\begin{array}
[c]{c}%
\ddot{a}(s)-\frac{\xi}{a(s)^{\kappa}}=0\\
a(0)=a_{0}>0,\text{ }\dot{a}(0)=a_{1},
\end{array}
\right.  \text{ } \label{eq124}%
\end{equation}
the solution exists locally.
\end{lemma}

After obtaining the nice structure of solutions (\ref{functional}), we just
use the techniques of separation of variable (\cite{GW}, \cite{M1}, \cite{Li},
\cite{Yuen1}, \cite{Yuen 2} and \cite{Yuen 3}), to prove the theorem:

\begin{proof}
[Proof of Theorem \ref{thm:1}]From Lemma \ref{lem:generalsolution}, it is
clear to see our functions (\ref{ChCh2}) fit well into the mass equation,
(\ref{2com})$_{1}$, except for two boundary points.

The second equation of 2-component Camassa-Holm equations (\ref{2com})$_{2}$,
becomes:%
\begin{align}
&  u_{t}-u_{xxt}+4uu_{x}-3u_{x}u_{xx}-uu_{xxx}+k_{3}\rho\rho_{x}\\
&  =u_{t}+4uu_{x}+k_{3}\rho\rho_{x}\label{abc}%
\end{align}
As the velocity $u$, in the solutions (\ref{ChCh2}) is a linear flow:
\begin{equation}
u=\frac{\dot{a}(4t)}{a(4t)}x
\end{equation}
we have
\begin{equation}
u_{xx}=0.
\end{equation}
The equation (\ref{abc}) becomes:
\begin{align}
&  =u_{t}+4u_{x}u+k_{3}\rho\rho_{x}\\[0.1in]
&  =\frac{\partial}{\partial t}\left(  \frac{\dot{a}(4t)}{a(4t)}\right)
x+4\left(  \frac{\dot{a}(4t)}{a(4t)}\right)  x\frac{\dot{a}(4t)}{a(4t)}%
+k_{3}\frac{f(\frac{x}{a(4t)^{k_{2}/4}})}{a(4t)^{(k_{1}+k_{2})/4}}\left(
\frac{f(\frac{x}{a(4t)^{k_{2}/4}})}{a(4t)^{(k_{1}+k_{2})/4}}\right)
_{x}\\[0.1in]
&  =\left(  4\frac{\ddot{a}(4t)}{a(4t)}-4\frac{\dot{a}(4t)^{2}}{a(4t)^{2}%
}\right)  x+4\left(  \frac{\dot{a}(4t)}{a(4t)}\right)  \frac{\dot{a}%
(4t)}{a(4t)}x+k_{3}\frac{f(\frac{x}{a(4t)^{k_{2}/4}})}{a(4t)^{(k_{1}+k_{2}%
)/4}}\frac{\dot{f}(\frac{x}{a(4t)^{k_{2}/4}})}{a(4t)^{(k_{1}+k_{2})/4}}%
\frac{1}{a(4t)^{k_{2}/4}}\\[0.1in]
&  =4\frac{\ddot{a}(4t)}{a(4t)}x+k_{3}\frac{f(\frac{x}{a(4t)^{k_{2}/4}}%
)\dot{f}(\frac{x}{a(4t)^{k_{2}/4}})}{a(4t)^{\frac{k_{1}}{2}+\frac{3k_{2}}{4}}%
}\\[0.1in]
&  =\frac{k_{3}}{a(4t)^{\frac{k_{1}}{2}+\frac{3k_{2}}{4}}}\left(  \frac{\xi
}{k_{3}}\eta+f(\eta)\dot{f}(\eta)\right)
\end{align}
for $k_{3}\neq0$, with the Emden equation:%
\begin{equation}
\left\{
\begin{array}
[c]{c}%
\ddot{a}(s)-\frac{\xi}{4a(s)^{\kappa}}=0\\
a(0)=a_{0}\neq0,\text{ }\dot{a}(0)=a_{1}%
\end{array}
\right.
\end{equation}
by defining the variables $s:=4t$, $\eta:=x/a(s)^{k_{2}/4}$ and $\kappa
=\frac{k_{1}}{2}+k_{2}-1.$\newline(1)for $k_{3}=0,$ we have $\xi=0:$%
\begin{equation}
\rho(t,x)=\frac{\rho_{0}(\eta)}{a(4t)^{(k_{1}+k_{2})/4}},\text{ }%
u(t,x)=\frac{\overset{\cdot}{a}(4t)}{a(4t)}x
\end{equation}
where $\rho_{0}$ is an arbitrary $C^{1}$ function.\newline Now, we can
separate the partial differential equations into two ordinary differential
equations. Then, we only need to solve for $\frac{\xi}{k_{3}}<0,$%
\begin{equation}
\left\{
\begin{array}
[c]{c}%
\frac{\xi}{k_{3}}\eta+f(\eta)\dot{f}(\eta)=0\\
f(0)=-\alpha\leq0
\end{array}
\right.  \label{ode1}%
\end{equation}
or for $\frac{\xi}{k_{3}}>0,$%
\begin{equation}
\left\{
\begin{array}
[c]{c}%
\frac{\xi}{k_{3}}\eta+f(\eta)\dot{f}(\eta)=0\\
f(0)=\alpha\geq0.
\end{array}
\right.  \label{ode2}%
\end{equation}
The ordinary differential equations (\ref{ode1}) or (\ref{ode2}) can be solved
exactly as
\begin{equation}
f(\eta)=\frac{k_{3}}{\xi}\sqrt{\frac{-\xi}{k_{3}}\eta^{2}+\left(  \frac{\xi
}{k_{3}}\alpha\right)  ^{2}}.
\end{equation}
In fact, we have the self-similar solutions in details:\newline(2) For
$k_{3}>0$ and $\xi>0$, or\newline(3) For $k_{3}<0$ and $\xi<0$:
\begin{equation}
\left\{
\begin{array}
[c]{c}%
\rho(t,x)=\max(\frac{f\left(  \eta\right)  }{a(4t)^{(k_{1}+k_{2})/4}},\text{
}0)\text{, }u(t,x)=\frac{\overset{\cdot}{a}(4t)}{a(4t)}x\\
\overset{\cdot\cdot}{a}(s)-\frac{\xi}{4a(s)^{\kappa}}=0,\text{ }a(0)=a_{0}%
\neq0,\text{ }\overset{\cdot}{a}(0)=a_{1}\\
f(\eta)=\frac{k_{3}}{\xi}\sqrt{\frac{-\xi}{k_{3}}\eta^{2}+\left(  \frac{\xi
}{k_{3}}\alpha\right)  ^{2}}.
\end{array}
\right.  \label{abc1}%
\end{equation}
With the assistance of Lemmas \ref{lemma22}, we may obtain the local existence
of the above solutions.

The proof is completed.
\end{proof}

After we construct the solutions (\ref{ChCh2}), the blowup or global behavior
can be analyzed from the Emden equation:%
\begin{equation}
\left\{
\begin{array}
[c]{c}%
\ddot{a}(s)-\frac{\xi}{a(s)^{\kappa}}=0\\
\text{ }a(0)=a_{0}\neq0,\text{ }\dot{a}(0)=a_{1}.
\end{array}
\right.
\end{equation}
To additionally show the blowup or global property of the time function
$a(s)$, the following lemma is needed. For the particular case of $k_{1}=1$
and $k_{2}=1$, we have
\begin{equation}
\kappa=\frac{k_{1}}{2}+k_{2}-1=\frac{1}{2}<1.
\end{equation}
For the case $\kappa=1$ or $\kappa=1/3$, the blowup or global results for the
Emden equation are already shown in \cite{Y1} and \cite{Yuen 3}. For the case
$0<\kappa\leq1$, we can show by the energy method \cite{LS}. Due to the proof
is very similar, we omit the details here to have the lemma:

\begin{lemma}
\label{lemma33 copy(1)}For $0<\kappa\leq1$, the Emden equation (\ref{Emden}),%
\begin{equation}
\left\{
\begin{array}
[c]{c}%
\ddot{a}(s)-\frac{\xi}{a(s)^{\kappa}}=0\\
a(0)=a_{0}>0,\text{ }\dot{a}(0)=a_{1},
\end{array}
\right.  \label{emden1/2}%
\end{equation}
(1) if $\xi<0$, there exists a finite time $S$, such that
\begin{equation}
\underset{s\rightarrow S^{-}}{\lim}a(s)=0.
\end{equation}
\newline(2) if $\xi>0$, the solution $a(t)$ exists globally, such that%
\begin{equation}
\underset{s\rightarrow+\infty}{\lim}a(s)=-\infty.
\end{equation}

\end{lemma}

After obtaining the above lemma, it is clear to have the following result:

\begin{corollary}
For $0<\kappa\leq1$ and $a_{0}>0$, we have:\newline(1) $k_{3}>0$ and $\xi>0$,
the solution (\ref{ChCh2}) exists globally.\newline(2) $k_{3}<0$ and $\xi<0$,
the solution (\ref{ChCh2}) blows up on a finite time $T$.
\end{corollary}

For the other cases, the interested reader may determine easily by the
classical energy method. In particular, for $\kappa>1$, the blowup or global
behavior may be referred \cite{DXY}.

\begin{remark}
For $k_{3}<0$ and $\xi<0$, the blowup solutions (\ref{ChCh2}) collapse at the
origin:%
\begin{equation}
\underset{t\rightarrow T^{-}}{\lim}\rho(t,0)=+\infty
\end{equation}
with a finite time $T;$\newline For $k_{3}>0$ and $\xi>0$, the global behavior
of the solution (\ref{ChCh2}) at the origin is:
\begin{equation}
\underset{t\rightarrow+\infty}{\lim}\rho(t,0)=0.
\end{equation}

\end{remark}

\begin{remark}
The solutions (\ref{ChCh2}) are only $C^{0}$ functions, as the function
$f(\eta)$ is discontinuous at the two boundary points, for $\alpha>0$:%
\begin{equation}
\underset{\eta^{2}\rightarrow\left\vert \xi\alpha\right\vert }{\lim}\dot
{f}(\eta)\neq0.
\end{equation}

\end{remark}

\begin{remark}
Our analytical solutions could provide concrete examples for testing the
validation and stabilities of numerical methods for the systems. Additionally,
our special solutions can shed some light on the understanding of evolutionary
pattern of the systems.
\end{remark}

\begin{remark}
We may calculate the mass of\newline(1) for $a_{0}>0$, the solutions
(\ref{ChCh2}) with $(k_{3}>0$ and $\xi>0)$ or $(k_{3}<0$ and $\xi<0)$:
\begin{equation}
Mass=\int_{-\infty}^{+\infty}\rho(t,x)dx<+\infty;
\end{equation}
\newline(2) for $a_{0}<0$ and $(-1)^{\kappa}=-1$, the solutions (\ref{ChCh2})
with $(k_{3}>0$ and $\xi<0)$ or $(k_{3}<0$ and $\xi>0)$ :%
\begin{equation}
Mass=+\infty.
\end{equation}

\end{remark}

\section{Blowup Phenomenon for $k_{3}\geq0$}

In this section, we show the blowup phenomenon by the standard method. The
system can be converted to a quasi-linear evolution of hyperbolic type:%
\begin{equation}
\left\{
\begin{array}
[c]{c}%
\rho_{t}=-k_{2}\rho_{x}u-(k_{1}+k_{2})\rho u_{x}\\
u_{t}+u\cdot u_{x}+\partial_{x}G\ast(\frac{3}{2}u^{2}+\frac{k_{3}}{2}\rho
^{2})=0
\end{array}
\right.  \label{2-comcon}%
\end{equation}
where the sign $\ast$ denotes the spatial convolution, $G(x)$ is the
associated Green function of the operator $(1-\partial_{x}^{2})^{-1}$. For the
local well-possness and the blowup phenomenon with odd initial values
$(\rho_{0},u_{0},)$, are discussed in \cite{JG} recently. We may show the
blowup by the particle trajectory method. Consider the following problem:%
\begin{equation}
\left\{
\begin{array}
[c]{c}%
q_{t}=u(q,t)\text{, }0<t<T\\
q(x,0)=x.
\end{array}
\right.
\end{equation}

we apply the characteristic curve method of hyperbolic partial differential
equations, as the following theorem:

\begin{theorem}
Suppose the velocity $u$ is uniformly bounded at some point $x_{0}$,$\ $such
that $u(t,x_{0})\leq M$ and $u_{x}(x_{0},0)<-\sqrt{\frac{3}{2}}\left\vert
M\right\vert $and the regularity of the solution at $x_{0}$ is sufficient
enough. The nontrivial solutions blow up before a finite time $T.$
\end{theorem}

\begin{proof}
In general, we show that the $\rho(t,x(t;x))$ preserves its positive nature as
the mass equation (\ref{2-comcon})$_{1}$ can be converted to be%
\begin{equation}
\rho_{t}+k_{2}\rho_{x}u=-(k_{1}+k_{2})\rho u_{x}%
\end{equation}%
\begin{equation}
\frac{D\rho}{Dt}+(k_{1}+k_{2})\rho\frac{\partial}{\partial x}\cdot u=0
\label{eqq2}%
\end{equation}
with the material derivative:%
\begin{equation}
\frac{D}{Dt}=\frac{\partial}{\partial t}+k_{2}\left(  u\frac{\partial
}{\partial x}\right)  . \label{eqq1}%
\end{equation}
We integrate the equation (\ref{eqq2})$:$%
\begin{equation}
\rho(t,x)=\rho_{0}(x_{0}(0,x_{0}))\exp\left(  -(k_{1}+k_{2})\int_{0}^{t}%
\nabla\cdot u(t,x(t;0,x_{0}))dt\right)  \geq0
\end{equation}
for $\rho_{0}(x_{0}(0,x_{0}))\geq0,$ along the characteristic curve.\newline
To drive the argument for blowing up, we differentiate equation
(\ref{2-comcon})$_{2}$ with respect to $x$:%
\begin{equation}
u_{xt}=-u_{x}^{2}-uu_{xx}+\frac{3}{2}u^{2}+\frac{k_{3}}{2}\rho^{2}%
-G\ast\left(  \frac{3}{2}u^{2}+\frac{k_{3}}{2}\rho^{2}\right)  .
\end{equation}
Along the another characteristic curve with
\begin{equation}
\frac{d}{dt}:=\frac{\partial}{\partial t}-u\frac{\partial}{\partial x},
\end{equation}
we have for some point $x_{0}$:%
\begin{align}
\frac{du_{x}(x_{0},t)}{dt}  &  =-u_{x}^{2}(x_{0},t)+\frac{3}{2}u^{2}%
(x_{0},t)-G\ast\left(  \frac{3}{2}u^{2}+\frac{k_{3}}{2}\rho^{2}\right)
(x_{0},t)\\
&  \leq-u_{x}^{2}(x_{0},t)+\frac{3}{2}M^{2}.
\end{align}
The above Racci equation blows up before a finite time $T$, if the initial
value,
\begin{equation}
u_{x}(x_{0},0)<-\sqrt{\frac{3}{2}}\left\vert M\right\vert .
\end{equation}
The proof is completed.
\end{proof}

\end{document}